\documentclass[reqno,fleqn]{amsart}
\usepackage{amsthm}
\usepackage{amssymb}
\usepackage{url}
\usepackage{enumitem}
\usepackage[hang,flushmargin]{footmisc}

\numberwithin{equation}{section}
\numberwithin{figure}{section}
\theoremstyle{plain}
\newtheorem{thm}{Theorem}[section]
\newtheorem*{thm*}{Theorem}

\newtheorem{prop}[thm]{Proposition}
\newtheorem{fact}[thm]{Fact}
\newtheorem{cor}[thm]{Corollary}

\newcommand{\eps}{\epsilon}

\newcommand{\zoo}{\{0,1\}}
\newcommand{\zoon}{\zoo^n}

\newcommand{\orth}{\operatorname{orth}}
\newcommand{\AND}{\text{\rm AND}}
\newcommand{\PTP}{\operatorname{PTP}}
\newcommand{\ED}{\operatorname{ED}}
\newcommand{\SURJ}{\operatorname{SURJ}}
\newcommand{\kED}{\operatorname{ED}^k}
\newcommand{\OR}{\operatorname{OR}}
\newcommand{\PP}{\mathsf{PP}}
\newcommand{\NISZK}{\mathsf{NISZK}}
\newcommand{\UPP}{\mathsf{UPP}}

\newcommand{\QMA}{\mathsf{QMA}}
\newcommand{\NP}{\mathsf{NP}}

\newcommand{\NIPZK}{\mathsf{NIPZK}}
\newcommand{\SBQP}{\mathsf{SBQP}}

\renewcommand{\Re}{\mathbb{R}}
\newcommand{\Domain}{\mathcal{D}}
\DeclareMathOperator{\dom}{dom}
\DeclareMathOperator{\im}{im}
\DeclareMathOperator*{\Exp}{\mathbf E}

\title{Vanishing-Error Approximate Degree and QMA Complexity}

\author[Alexander A. Sherstov]{Alexander A. Sherstov$^*$} 
\author[Justin Thaler]{Justin Thaler$^\dagger$}

\thanks{$^{*}$ Computer Science Department, UCLA, Los Angeles, CA~90095.
Supported by NSF grant CCF-1814947.
 Email: \texttt{sherstov@cs.ucla.edu}\\
 $^\dagger$ Department of Computer Science,  Georgetown University, Washington, DC~20057.
Supported by NSF CAREER award CCF-1845125.  Email: \texttt{justin.thaler@georgetown.edu}
 }

\begin{document}

\begin{abstract}
The \emph{$\epsilon$-approximate degree} of a function $f\colon X \to \zoo$
is the least degree of a multivariate real polynomial $p$ such that $|p(x)-f(x)| \leq \epsilon$ for all $x \in X$.
We determine the $\eps$-approximate degree of the element distinctness function, the surjectivity function, and the permutation testing problem,
showing they are $\Theta(n^{2/3}
\log^{1/3}(1/\epsilon))$, $\tilde\Theta(n^{3/4} \log^{1/4}(1/\epsilon))$, and 
$\Theta(n^{1/3} \log^{2/3}(1/\epsilon))$, respectively. Previously, these bounds were known only for constant~$\epsilon.$ 

We also derive a connection between vanishing-error approximate
degree and quantum Merlin--Arthur (QMA) query complexity.
We use this connection to show that the QMA complexity of permutation testing is
$\Omega(n^{1/4})$. This
improves on the previous best lower bound of
$\Omega(n^{1/6})$ due to Aaronson (\emph{Quantum Information \& Computation}, 2012),
and comes somewhat close to matching a known upper bound of $O(n^{1/3})$.

\end{abstract}

\maketitle

\belowdisplayskip=12pt plus 1pt minus 3pt 
\abovedisplayskip=12pt plus 1pt minus 3pt 

\thispagestyle{empty}

\section{Introduction}
The \emph{$\epsilon$-approximate degree} of a function $f\colon X \to \zoo$, 
denoted $\deg_{\epsilon}(f)$, is the least degree of a multivariate real-valued polynomial $p$ such that $|p(x)-f(x)| \leq \epsilon$ for all inputs $x \in X$.
Lower bounds on approximate degree have many applications in theoretical computer science,
ranging from quantum query and communication lower bounds, to oracle separations and cryptographic secret sharing schemes. 
Upper bounds on approximate degree have important algorithmic implications in learning theory and 
differential privacy, and underlie state-of-the-art circuit and formula size lower bounds. 
The interested reader can find a bibliographic overview of these applications 
in~\cite{bkt, sherstovalgopoly}.

This paper focuses on three well-studied functions
whose approximation by polynomials has 
applications to quantum computing and beyond.
The first function is
\emph{element distinctness} $\ED_n,$ where the input is
a list of $n$ numbers from $\{1,2,\ldots,n\}$ and the
objective is to determine if the numbers are pairwise
distinct. The second function is
\emph{surjectivity} $\SURJ_{n,r},$ where the input is a
list of $n$ numbers from the range $\{1,2,\ldots,r\}$ and the
goal is to check whether every range element appears on
the list. The canonical setting is $r=\lfloor
cn\rfloor$ for some constant $0<c<1.$ The third problem
that we study is \emph{permutation testing}
$\PTP_{n,\alpha},$ parameterized by a constant
$0<\alpha<1.$ Here, the input is 
a list of $n$ numbers from $\{1,2,\ldots,n\},$ and the
objective is to distinguish the case when 
the list contains every range element from the case
when the list contains at most $\alpha n$ range elements. 
In the context of polynomial approximation, it is
customary to represent the input to
these functions as a Boolean
matrix $x=[x_{i,j}]$, where $x_{i,j}=1$ if and only if
the $i$th element on the list equals $j$.

\subsection*{Vanishing-error approximate degree}
Much work in the area has focused on \emph{bounded-error} approximate degree, defined
for a Boolean function $f$ as the quantity $\deg_{1/3}(f).$ 
The choice of constant $1/3$ here is arbitrary,
as $\deg_{\eps}(f)=\Theta(\deg_{1/3}(f))$ for all constants $0<\eps<1/2$. 
In particular, the bounded-error approximate degrees of 
element distinctness,
surjectivity, and
permutation testing are known to be
$\Theta(n^{2/3}),$ 
$\tilde\Theta(n^{3/4}),$ 
and 
$\Theta(n^{1/3}),$ respectively~\cite{%
aaronsonshi,
ambainis05collision,
kutin05collision,
aaronson,
sherstovalgopoly,
bkt}.
Our understanding of approximate degree
with vanishing error, $\epsilon=o(1),$ is far less complete.
Among the very few functions whose vanishing-error 
approximate degree has been determined is the $n$-bit $\AND$
function, with the asymptotic bound
$\deg_{\eps}(\AND_n) = \Theta(n^{1/2} \log^{1/2}(1/\eps))$
due to Buhrman et al.~\cite{bCdWZ}.
We give a new and entirely different proof of their result.
Our technique further allows us to settle the vanishing-error
approximate degrees of the much more complicated
functions of element distinctness, surjectivity, and
permutation testing:

\begin{thm}
\label{thm:vanishing}
Let $0<c<1$ and $0<\alpha<1$ be arbitrary constants.
Then
\begin{align*}
&\deg_\epsilon(\ED_n)=\Omega\left(n^{2/3} \left(\log\frac1\eps\right)^{1/3}\right),\\
&\deg_\epsilon(\SURJ_{n,\lfloor cn\rfloor})=\tilde\Omega\left(n^{3/4} \left(\log\frac1\eps\right)^{1/4}\right),\\
&\deg_\epsilon(\PTP_{n,\alpha})=\Omega\left(n^{1/3} \left(\log\frac1\eps\right)^{2/3}\right)
\end{align*}
for all $1/3^n\leq \epsilon\leq 1/3.$ 
\end{thm}
\noindent
This theorem is optimal with respect to all parameters. The lower bounds 
for element distinctness and surjectivity match the
vanishing-error constructions in~\cite{sherstovalgopoly}, whereas the lower
bound for permutation testing is tight by a quantum
query argument which we include as
Theorem~\ref{thm:ptp-qquery-upper}.
A comment is in order on $\epsilon$-approximate
degree in the complementary range,
$\epsilon<1/3^n.$
Routine interpolation gives an exact
representation for each of the functions in Theorem~\ref{thm:vanishing} 
as a polynomial of degree at most $n.$
Theorem~\ref{thm:vanishing} shows that this upper bound
is asymptotically tight, settling the
$\epsilon$-approximate degree for $\epsilon<1/3^{n}$ as
well.

We prove a result analogous to
Theorem~\ref{thm:vanishing}
for \emph{$k$-element distinctness} $\kED_n$, a well-studied
generalization of $\ED_n$.
Specifically, we prove that if $\kED_n$ has bounded-error approximate degree $\Omega(n^{\ell})$,
then it has $\epsilon$-approximate degree $\Omega(n^{\ell} \log^{1-\ell}(1/\eps))$.
The best known lower bound on the bounded-error approximate degree of $\kED_n$
is $\tilde\Omega(n^{3/4-1/(2k)})$ \cite{bkt}, so this yields 
\[\deg_\eps(\kED_n)=\tilde\Omega\left(n^{\frac34-\frac1{2k}} \left(\log\frac1\eps\right)^{\frac14+\frac1{2k}}\right).\]
For large $k$, this comes 
close to the best known upper bound~\cite{sherstovalgopoly}:
\[\deg_\eps(\kED_n)=O\left(n^{\frac34-\frac{1}{4(2^k-1)}} \left(\log\frac1\eps\right)^{\frac14+\frac{1}{4(2^k-1)}}\right).\]

Our techniques are quite general, 
and we are confident that they will find other
applications in the area. 
The technical core of our results establishes that for 
any function $f_n$ that  contains $\AND_k \circ f_{\lfloor n/k \rfloor}$ as a subfunction for each $k \leq n$, any bounded-error
approximate degree lower bound for $f_n$
automatically implies a strong lower bound for the $\epsilon$-approximate degree of $f_n$. 
This allows us to prove tight lower bounds on the
vanishing-error approximate degrees of
$\AND_n,$ $\ED_n,$ $\kED_n,$ and $\SURJ_{n,r}$. To handle
$\PTP_{n,\alpha},$ we generalize our technique to other outer
functions.
Our analysis is based on the so-called method of dual polynomials,
whereby one proves approximate degree lower bounds
by constructing explicit dual solutions to a certain linear program
capturing the approximate degree of the given function.

In the remainder of the introduction, we focus on an application of Theorem~\ref{thm:vanishing} to
quantum Merlin--Arthur complexity.

\subsection*{The Merlin--Arthur model}
The Merlin--Arthur (MA) model of query complexity features a function $f$ and two asymmetric players,
Merlin and Arthur.  Arthur's goal is to compute 
$f$ on some unknown input $x$ while querying as few bits of
$x$ as possible. Merlin, who  knows $x$,
can help Arthur compute $f(x)$ by sending him a single witness, i.e.,
an arbitrary message of some bit length $m$. However, Merlin is untrusted. 
The model requires that, for any $x \in f^{-1}(1)$, there is some Merlin message causing Arthur to output 1 with
probability at least $2/3$, and for any $x \in f^{-1}(0)$, no 
Merlin message can cause Arthur to output 1 with
probability more than $1/3$. The cost of the protocol is the sum of the witness length $m$ and 
the number of bits of $x$ queried by Arthur. 
In \emph{quantum} Merlin-Arthur (QMA) query complexity, the witness
sent by Merlin is allowed to be an arbitrary $m$-qubit quantum message,
and Arthur is permitted to query bits of the input $x$ in superposition. 
The MA and QMA query models have important analogues in communication complexity and
Turing machine complexity. In the former setting,
Arthur is replaced by two parties Alice and Bob, and the input $x$
is split between them.

The complexity class $\QMA$ is a quantum analog of $\NP$
and accordingly has received considerable attention.
It is well known that any $\QMA$ protocol can be 
simulated by an $\SBQP \subseteq \PP$ protocol with at most a quadratic blowup in cost,
i.e.,  $\QMA(f) \geq \Omega(\SBQP(f)^{1/2})$
\cite{vyalyi2003qma}.\footnote{An SBQP protocol $\mathcal{A}$
is a quantum protocol for which there is some $\alpha$
such that $\mathcal{A}$ accepts every input
in $f^{-1}(1)$ with probability at least $\alpha$, and every
input in $f^{-1}(0)$ with probability
at most $\alpha/2$ \cite{kuperberg2009hard}.} 
In turn, the existence of an SBQP query protocol 
that makes at most $c$ queries
implies that the \emph{one-sided} approximate degree of $f$ 
is at most $O(c)$. Here, the one-sided $\epsilon$-approximate degree of $f$
is the least degree of a real polynomial $p$
such that $|p(x)| \leq \eps$ for all $x \in f^{-1}(0)$, and $p(x) \geq 1-\eps$ for all $x \in f^{-1}(1)$ \cite{BT15}.
As a consequence, one can prove QMA query lower bounds for $f$ by lower bounding
the one-sided approximate degree of $f$.

Only a handful of additional results are known about QMA query and communication complexity. 
Raz and Shpilka~\cite{razshpilka} showed that $\AND_n$ has QMA query complexity $\Theta(\sqrt{n})$.
Klauck~\cite{klauck} showed that the QMA communication complexity
 of the disjointness problem is $\Omega(n^{1/3})$. 
 Neither of these results follows from a na\" ive application of the bound $\QMA(f) \geq \Omega(\sqrt{\SBQP(f)})$.

\subsection*{QMA complexity of permutation testing}
The permutation testing problem $\PTP_{n,\alpha}$ has played an important role in the study of interactive proof systems because
it possesses a simple non-interactive perfect zero knowledge  ($\NIPZK$)  protocol of logarithmic cost,
yet is a hard problem in many other models. 
Hence, it has been used to prove a variety of complexity class separations.
In particular, Aaronson~\cite{aaronson} showed that 
the QMA query complexity of $\PTP_{n,\alpha}$ is $\Omega(n^{1/6})$, and thereby gave an oracle 
 separating $\NIPZK$ from $\QMA$. 
Bouland et al.~\cite{bchtv} built on Aaronson's result to give an oracle separating non-interactive \emph{statistical} zero knowledge  ($\NISZK$)
from the complexity class $\UPP$, answering a question of Watrous from 2002. 
Gur, Liu, and Rothblum~\cite{gr} showed that the MA query complexity
of $\PTP_{n,\alpha}$ is $\Omega(n^{1/4})$. 
Despite this progress, the precise QMA complexity of $\PTP_{n,\alpha}$ 
has remained open, with the best upper bound being $O(n^{1/3})$ \cite{bht, aaronson} and the best
lower bound being Aaronson's $\Omega(n^{1/6})$. We
obtain a polynomially stronger lower bound.

\begin{thm}
\label{thm:qma}
Let $0<\alpha<1$ be an arbitrary constant. Then any QMA query 
protocol for $\PTP_{n,\alpha}$ with 
witness length $m$ has query cost $\Omega(n/m)^{1/3}$.
In particular, $\PTP_{n,\alpha}$ has QMA complexity $\Omega(n^{1/4}).$
\end{thm}
\noindent
This result quantitatively matches the MA lower bound of Gur et al.~\cite{gr}
but holds in the more powerful quantum setting.
Theorem~\ref{thm:qma} comes reasonably close
to matching the known QMA query upper bound of $O(n^{1/3}),$ which holds even
if Merlin does not send any message to Arthur; see Theorem~\ref{thm:ptp-qquery-upper}.

To prove Theorem~\ref{thm:qma}, 
we derive a connection between QMA query complexity and vanishing-error approximate
degree for a class of functions that includes $\AND_n$, $\ED_n$, and $\PTP_{n,\alpha}$. This connection
amounts to the observation that, for these particular functions, the one-sided $\eps$-approximate
degree is \emph{equal} to the $\eps$-approximate degree.
Prior work on QMA complexity (e.g., \cite{klauck}) has implicitly exploited a similar observation
in the special case of $\AND_n$. 
Our analysis substantially generalizes the insights of prior work,
and makes explicit the key phenomenon at play, namely the equivalence of one-sided vs.~standard approximate degree
for these functions.
Combining this connection with our new vanishing-error approximate
degree lower bounds in Theorem~\ref{thm:vanishing} establishes Theorem~\ref{thm:qma}.

\section{Preliminaries}

\label{s:defs}

For a function $f,$ we let $\dom f$ and $\im f$ stand for the
domain and image of $f,$ respectively.
We view Boolean functions as mappings $f\colon X\to\zoo$
for a finite set $X.$
For functions 
$g \colon X \to Y$ and $f \colon Y^n \to Z$,
we let $f \circ g$ denote the block-composition
of $f$ and $g$. In more detail, $f \circ g\colon
X^n\to Z$ is the function that maps $(x_1, \dots, x_n) \in X^n$
to $f(g(x_1), \dots, g(x_n))$. We generalize
block-composition to the case when the domain of $f$ is
properly contained in $Y^n$ by defining the domain of $f\circ g$ 
as the set of $(x_1,\ldots,x_n)\in X^n$ such that
$(g(x_1),\ldots,g(x_n))\in \dom f.$

\subsection{Polynomial approximation}
For a multivariate real polynomial $p\colon\Re^{n}\to\Re$, we let
$\deg p$ denote the total degree of $p$, i.e., the largest degree
of any monomial of $p.$ It will be convenient to define
the degree of the zero polynomial by $\deg0=-\infty.$ For a real-valued
function $\phi$ supported on a finite subset of $\Re^{n}$, we define
the \emph{orthogonal content of $\phi,$} denoted $\orth\phi$, to
be the minimum degree of a real polynomial $p$ for which $\langle\phi,p\rangle\ne0.$
We adopt the convention that $\orth\phi=\infty$ if no such polynomial
exists. 
For two functions $f, \psi \colon X \to \Re$,
let $\langle f, \psi \rangle = \sum_{ x \in X} f(x) \psi(x)$ denote
the correlation of $f$ and $\psi$, and let $\|\psi\|_1 = \sum_{x \in X} |\psi(x)|$.
For any real-valued
function $\psi \colon X \to \Re$, its $k$-th tensor
power $\psi^{\otimes k}\colon X^k \to \Re$ is
given by $\psi^{\otimes k}(x_1, \dots, x_k) = \psi(x_1) \cdots \psi(x_k)$.

The \emph{$\eps$-approximate degree} of a function $f \colon X \to \Re$, denoted
$\deg_{\eps}(f)$, is the least degree of a polynomial $p \colon X \to \Re$
such that $|p(x)-f(x)| \leq \eps$ for all $x \in X$.
We emphasize that no 
restriction is placed on the behavior of $p$ at inputs
outside $f$'s domain of definition, $X.$ For most functions
of interest to us, the domain $X$ is a proper subset
of $\zoon$ and thus their approximating polynomials may take on
arbitrary values on $\zoon\setminus X.$
The following dual characterization of approximate
degree is well known and can be verified using 
linear programming duality.

\begin{fact} \label{fact:duality}
Fix $d>0$ and a function $f \colon X \to \Re$. Then
$\deg_{\eps}(f) \geq d$ if and only if
there exists a function $\psi \colon X \to \Re$
such that 
\begin{align*}
&\langle f,\psi\rangle>\epsilon\|\psi\|_{1},\\
&\orth\psi \geq d.
\end{align*}
\end{fact}

The simplest function of interest to us is 
$\AND_n \colon \zoo^n \to \zoo,$ given as usual by 
$\AND_n(x)=\bigwedge_{i=1}^n x_i.$ 
Its bounded-error approximate degree 
was determined by
Nisan and Szegedy~\cite{NS94}, as follows.
\begin{thm}
\label{thm:and-bounded}
For all $n\geq1,$
\[\deg_{1/3}(\AND_n)=\Theta(\sqrt n).\]
\end{thm}

\subsection{Surjectivity}
Let $\Domain_{n,r}$ stand for the set of Boolean
matrices of size $n\times r$ in which every row has
exactly one $1.$ 
Every matrix $x\in\Domain_{n,r}$ has a natural
interpretation as specifying a mapping
$\phi \colon \{1,2,\dots, n\} \to \{1,2,\dots, r\}$,
where $\phi(i)=j$ if and only if $x_{i,j}=1.$ Our next
three functions are defined on $\Domain_{n,r}$ and can
thus be regarded as ``function properties."
To start with, the \emph{surjectivity problem} with $n$ elements and
range size $r$ is defined as 
$\SURJ_{n,r}\colon \Domain_{n,r}\to\zoo,$
where 
\[\SURJ_{n,r}(x)=\bigwedge_{j=1}^r\bigvee_{i=1}^n
x_{i,j}.\]
Thus, 
$\SURJ_{n,r}$ takes as input an $n\times r$ Boolean matrix in which every row
contains exactly one $1$, and 
evaluates to 1 if and only if every column
of the input contains at least one 1. 
Interpreting the input matrix as a mapping, 
$\SURJ_{n, r}$ evaluates to $1$ if and only if that
mapping is surjective.
This surjectivity property is trivially false for $r>n,$ and the standard
setting of parameters is $r=\lfloor cn\rfloor$ for some constant
$0<c<1.$ The choice of constant $c$ is
unimportant because it affects
$\deg_{1/3}(\SURJ_{n,\lfloor cn\rfloor})$ by at most a
multiplicative constant. It
was shown in~\cite{sherstovalgopoly} that
the surjectivity function has bounded-error approximate
degree $O(n^{3/4}).$ Bun et
al.~\cite{bkt} gave an alternate proof of this
upper bound and additionally proved that it is tight up to
a polylogarithmic factor. We thus have:
\begin{thm}
\label{thm:surj-bounded}
Let $0<c<1$ be an arbitrary constant. Then 
\[\deg_{1/3}(\SURJ_{n,\lfloor cn\rfloor })=\tilde\Theta(n^{3/4}).
\]
\end{thm}

\subsection{Element distinctness}
Another well-studied function is \emph{element
distinctness} $\ED_{n,r}\colon\Domain_{n,r}\to\zoo,$
defined by
$\ED_{n,r}(x)=1$ if and only if every column of the
input matrix $x$ has at most one $1.$ Switching to the
interpretation of $x$ as a mapping, $\ED_{n,r}(x)$ 
evaluates to true if and only if the mapping is
one-to-one. This property is trivially false for
$r<n.$ In the complementary case,
Ambainis~\cite{ambainis05collision} proved that for any given $\epsilon,$ the
$\epsilon$-approximate degree of $\ED_{n,r}$ is the
same for all $r\geq n.$ This means that one may without
loss of generality focus on the special case $r=n,$
with the shorthand notation $\ED_{n}=\ED_{n,n}$.
Aaronson and Shi~\cite{aaronsonshi},
Ambainis~\cite{ambainis05collision},
and Kutin~\cite{kutin05collision} showed 
that $\ED_{n}$ has bounded-error approximate degree
$\Theta(n^{2/3}).$
\begin{thm}
\label{thm:ed-bounded}
For all $n\geq1,$
\[\deg_{1/3}(\ED_n)=\Theta(n^{2/3}).
\]
\end{thm}
\noindent
Element distinctness generalizes in a natural way to 
a function called \emph{$k$-element distinctness},
denoted $\kED_{n}\colon\Domain_{n,n}\to\zoo$.
This new function evaluates to true if and only if
the input matrix has no column with $k$ or more 1s.
Viewing the input as a mapping, $\kED_{n}$ evaluates to true if
and only if no range element occurs $k$ or more times.
With these definitions, we have $\ED_n=\ED_n^2.$

\subsection{Permutation testing}
The final problem of interest to us
is a restriction of element distinctness $\ED_n$.
In more detail, fix an integer $n\geq1$ and a real
number $0<\alpha<1$.
The domain of the \emph{permutation testing problem} $\PTP_{n,\alpha}$ is the set of all 
matrices $x\in\Domain_{n,n}$ in which 
the number of columns containing a
$1$ is either exactly $n$ 
or at most $\alpha n.$ The function
evaluates to true in the former case and to false in
the latter. Equivalently, $\PTP_{n,\alpha}(x)=1$ if and only if $x$
is a permutation matrix. In the regime of interest to
us, $0<\alpha<1$ is a constant independent of $n.$ 

The permutation testing problem was introduced by
Aaronson~\cite{aaronson}, who defined it somewhat
differently. In his variant of permutation testing,
which we denote by $\PTP^*_{n,\delta},$ one
is given a matrix $x\in\Domain_{n,n}$ that is either
(i)~a permutation matrix, or (ii)~disagrees from every
permutation matrix in at least $\delta n$ 
rows. The function evaluates to true in case~(i) and to
false in case~(ii). 
As the following proposition shows, Aaronson's
$\PTP^*_{n,\delta}$ is precisely the same function as
our $\PTP_{n,1-\delta}.$

\begin{prop}
\label{prop:ptp-ptp}
Let $0<\delta<1$ and $n\geq1$ be given. Then as functions,
\begin{align*}
\PTP^*_{n,\delta}=\PTP_{n,1-\delta}.
\end{align*}
Specifically, the l.h.s.~and r.h.s.~have the same
domain and agree at every point thereof.
\end{prop}

\begin{proof}
This claim is easiest to verify by interpreting an
input $x\in\Domain_{n,n}$ as a mapping
$\phi\colon\{1,2,\ldots,n\}\to\{1,2,\ldots,n\}.$
A moment's reflection shows that $\phi$ disagrees from
every permutation $\{1,2,\ldots,n\}\to\{1,2,\ldots,n\}$
in at least $n-|\im \phi|$ points, and there is a
permutation that achieves this lower bound. Restating
this in matrix terminology, 
a matrix $x\in\Domain_{n,n}$ disagrees from every
permutation matrix in at least $\delta n$ rows if and
only if the number of columns of $x$ containing a $1$
is at most $n-\delta n.$
\end{proof}

By adapting earlier analyses of element distinctness,
Aaronson~\cite{aaronson} obtained the following result.
\begin{thm}
\label{thm:ptp-bounded}
Let $0<\delta<1$ be an arbitrary constant. Then 
\[\deg_{1/3}(\PTP^*_{n,\delta})=\Omega(n^{1/3}).
\]
\end{thm}
\noindent
This result is stated in~\cite{aaronson} specifically for
$\delta=1/8,$ but the proof actually allows any $0<\delta<1.$
Combining this theorem
with~Proposition~\ref{prop:ptp-ptp} gives the following
corollary.
\begin{cor}\label{cor:ptp-bounded}
Let $0<\alpha<1$ be an arbitrary constant. Then 
\[\deg_{1/3}(\PTP_{n,\alpha})=\Omega(n^{1/3}).
\]
\end{cor}

We close this section with a remark on input
encoding. In this work, functions like $\SURJ_{n, r}$
take as input a Boolean matrix $x$
in which every row has exactly one~1. 
Some other works~\cite{beame, bkt} represent the
input as a list $y_1, \dots, y_n\in\zoo^{\lceil \log
r\rceil},$ 
where $y_i$ 
encodes the location of the unique 1 in the
$i$-th row of the matrix representation $x.$ 
Switching to this alternate representation affects the approximate degree 
by at most a logarithmic factor. See~\cite{sherstovalgopoly} for
a detailed treatment of the relationship between these
representations.

\section{Approximate Degree Lower Bounds}

In this section, we study the vanishing-error
approximate degree of element distinctness,
surjectivity, and permutation testing, and in
particular settle Theorem~\ref{thm:vanishing} from the introduction.
The core of our technique is the following auxiliary
result.

\begin{prop}
\label{prop:tensor-deg-eps}For any $\epsilon\geq0$ and any function
$f\colon X\to\Re$ on a finite subset $X$ of Euclidean space,
\begin{align*}
\deg_{\epsilon^{k}}(f^{\otimes k}) & \geq k\deg_{\epsilon}(f), &  & k=1,2,3,\ldots.
\end{align*}
In particular, every function $f\colon X\to\zoo$ satisfies
\begin{align*}
\deg_{\epsilon^{k}}(\AND_{k}\circ f) & \geq k\deg_{\epsilon}(f), &  & k=1,2,3,\ldots.
\end{align*}
\end{prop}

\begin{proof}
We may assume that $\deg_\epsilon(f)\ne0$ since the
proposition is trivial otherwise. 
Let $\psi$ be an $\epsilon$-error dual polynomial for $f$,
as guaranteed by Fact \ref{fact:duality}:
\begin{align*}
 & \langle f,\psi\rangle>\epsilon\|\psi\|_{1},\\
 & \orth\psi=\deg_{\epsilon}(f).
\end{align*}
Then
\begin{align*}
\langle f^{\otimes k},\psi^{\otimes k}\rangle & =\langle f,\psi\rangle^{k}\\
 & >(\epsilon\|\psi\|_{1})^{k}\\
 & =\epsilon^{k}\|\psi^{\otimes k}\|_{1}.
\end{align*}
Applying Fact~\ref{fact:duality} once again,
\begin{align*}
\deg_{\epsilon^{k}}(f^{\otimes k}) & \geq\orth\psi^{\otimes k}\\
 & =k\orth\psi\\
 & =k\deg_{\epsilon}(f). &  & \qedhere
\end{align*}
\end{proof}

\noindent
The proof of Proposition~\ref{prop:tensor-deg-eps}
applies more generally to the conjunction of
$k$ distinct functions, but we will not need this
generalization.

\subsection{Warmup}
To illustrate our technique in the simplest possible
setting, we consider the well-studied $\AND_n$
function. Buhrman et al.~\cite{bCdWZ} proved that its 
$\epsilon$-error approximate degree is
$\Theta(\sqrt{n\log(1/\epsilon)}).$ We give a new and
simple proof of their lower bound.

\begin{thm}
For all $1/3^n\leq\epsilon\leq1/3,$
\begin{align}
\deg_{\epsilon}(\AND_{n}) & =\Omega\left(\sqrt{n\log\frac{1}{\epsilon}}\right). &  &\label{eq:and-deg-eps}
\end{align}
\end{thm}

\begin{proof}
For $k=1,2,\ldots,n,$  we have
\begin{align*}
\deg_{3^{-k}}(\AND_{n}) & \geq\deg_{3^{-k}}(\AND_{k}\circ\AND_{\lfloor n/k\rfloor})\\
 & \geq k\deg_{1/3}(\AND_{\lfloor n/k\rfloor})\\
 & =k\cdot\Omega\left(\sqrt{\frac{n}{k}}\right)\\
 & =\Omega(\sqrt{nk}),
\end{align*}
where the first, second, and third steps use 
the identity $\AND_{n_{1}n_{2}}=\AND_{n_{1}}\circ\AND_{n_{2}},$ Proposition~\ref{prop:tensor-deg-eps},
and Theorem~\ref{thm:and-bounded}, respectively.  This directly implies~(\ref{eq:and-deg-eps}).
\end{proof}

\subsection{Element distinctness}
Our next result is a tight lower bound on the
vanishing error approximate degree of element
distinctness, matching the 
upper bound from~\cite{sherstovalgopoly}. 

\begin{thm} \label{edthm}
For all $1/3^n\leq\epsilon\leq1/3,$
\begin{align}
\deg_{\epsilon}(\ED_{n}) & =\Omega\left(n^{2/3}\left(\log\frac{1}{\epsilon}\right)^{1/3}\right). \label{eq:ed-deg-eps}
\end{align}
\end{thm}

\begin{proof}
For any $k=1,2,3,\ldots,n$, we claim that $\AND_{k}\circ\ED_{\lfloor n/k\rfloor}$
is a subproblem of $\ED_{n}$. To see why, recall that the input
to $\ED_{n}$ is an $n\times n$ Boolean matrix in which every row
$i$ contains exactly one $1$, corresponding to the value of the
$i$th element. Now, fix $k\in\{1,2,\ldots,n\}$ and consider the
restriction of $\ED_{n}$ to input matrices that are \emph{block-diagonal,}
with $k$ blocks of size $\lfloor n/k\rfloor$ each and an additional
block of $n-k\lfloor n/k\rfloor$ ones on the diagonal. Each of the
first $k$ blocks corresponds to an instance of $\ED_{\lfloor n/k\rfloor},$
and the overall problem amounts to computing the AND of these $k$
instances. Therefore, $\AND_{k}\circ\ED_{\lfloor n/k\rfloor}$
is a subproblem of $\ED_{n}$, and
\begin{equation}
\deg_{\epsilon}(\ED_{n})\geq\deg_{\epsilon}(\AND_{k}\circ\ED_{\lfloor n/k\rfloor})\label{eq:ed-tensor-subproblems}
\end{equation}
for all $\epsilon$ and all $k=1,2,3,\ldots,n.$

The rest of the proof is closely analogous to that for $\AND_{n}.$
For $k=1,2,\ldots,n,$ 
\begin{align*}
\deg_{3^{-k}}(\ED_{n}) & \geq\deg_{3^{-k}}(\AND_{k}\circ\ED_{\lfloor n/k\rfloor})\\
 & \geq k\deg_{1/3}(\ED_{\lfloor n/k\rfloor})\\
 & \geq k\cdot\Omega\left(\frac{n}{k}\right)^{2/3}\\
 & =\Omega(n^{2/3}k^{1/3})
\end{align*}
where the first three steps use (\ref{eq:ed-tensor-subproblems}),
Proposition~\ref{prop:tensor-deg-eps}, and
Theorem~\ref{thm:ed-bounded}, respectively.
This directly implies~(\ref{eq:ed-deg-eps}).
\end{proof}

The previous proof shows more
generally that $\AND_k\circ\ED^r_{\lfloor n/k\rfloor}$
is a subfunction of $\ED^r_n$ for any
$k=1,2,\ldots,n.$ As a result, our analysis 
of element distinctness proves the following statement.

\begin{thm} \label{ked}
Fix constants $r\geq2$ and $\ell \in [0, 1]$  such that 
\[\deg_{1/3}(\ED^r_{n}) = \Omega(n^{\ell}).\]
Then 
\begin{align*}
\deg_{\epsilon}(\ED^r_{n}) & =\Omega\left(n^{\ell}\left(\log\frac{1}{\epsilon}\right)^{1-\ell}\right), &  & \frac{1}{3^{n}}\leq\epsilon\leq\frac{1}{3}. 
\end{align*}
\end{thm}
\noindent
Combining Theorem~\ref{ked} with the known lower bound 
\[\deg_{1/3}\left(\ED^r_{n}\right) = \tilde\Omega\left(n^{\frac34-\frac1{2r}}\right)\]
due to~\cite{bkt},
we conclude that 
\[\deg_{\epsilon}\left(\ED^r_{n}\right) =
\tilde\Omega\left(n^{\frac34-\frac1{2r}} \left(\log\frac1\epsilon\right)^{\frac14 + \frac1{2r}}\right)
\]
for $1/3^{n}\leq\epsilon\leq 1/3.$
Moreover, Theorem \ref{ked} will, in a black-box
manner, translate any future improvement in the
bounded-error lower bound for $\ED^r_{n}$
into an improved vanishing-error lower bound.

\subsection{Surjectivity}
An instance $x$ of the surjectivity problem
$\SURJ_{n,r}$ can be embedded
inside a larger instance of surjectivity in many ways,
e.g., by duplicating a row of $x$ or
by forming a block-diagonal matrix with blocks $x$
and $1.$ These two transformations yield
\begin{align}
&\deg_\epsilon(\SURJ_{n,r})\leq \deg_\epsilon(\SURJ_{n+1,r}),
\label{eq:surj-extend-domain}\\
&\deg_\epsilon(\SURJ_{n,r})\leq \deg_\epsilon(\SURJ_{n+1,r+1}),
\label{eq:surj-extend-domain-and-range}
\end{align}
respectively.
We will now prove an essentially tight lower bound on the
vanishing-error approximate degree of surjectivity,
matching the upper bound from~\cite{sherstovalgopoly}
up to a logarithmic factor.

\begin{thm} \label{thm:surj}
Let $0<c<1$ be an arbitrary constant. Then 
\begin{align*}
\deg_{\epsilon}(\SURJ_{n,\lfloor cn\rfloor })=
\tilde\Omega\left(n^{3/4}\left(\log\frac{1}{\epsilon}\right)^{1/4}\right),&&
\frac{1}{3^{n}}\leq\epsilon\leq\frac{1}{3}. 
\end{align*}
\end{thm}

\begin{proof}
The proof is a cosmetic adaptation of the analysis of element distinctness.
To start with, we claim that for any positive integers $n,r,k$ such that $k\mid n$ and
$k\mid r,$
the composition $\AND_{k}\circ\SURJ_{n/k,r/k}$
is a subproblem of $\SURJ_{n,r}$. Indeed, the input
to $\SURJ_{n,r}$ is an $n\times r$ Boolean matrix in which every row
$i$ contains exactly one $1$.
Consider the
restriction of $\SURJ_{n,r}$ to input matrices that are block-diagonal,
with $k$ blocks of size $n/k\,\times\, r/k$ each.
Each of these
blocks corresponds to an instance of $\SURJ_{n/k,r/k},$
and the overall problem amounts to computing the AND of these $k$
instances. This settles the claim.

Now let $n$ be arbitrary. Then for all positive
integers $k\leq\min\{cn,(1-c)n\},$ 
\begin{align*}
\deg_{3^{-k}}(\SURJ_{n,\lfloor cn\rfloor})
 & \geq\deg_{3^{-k}}(\SURJ_{n-(\lfloor cn\rfloor-k\lfloor cn/k\rfloor),k\lfloor cn/k\rfloor})\\
 & \geq\deg_{3^{-k}}(\SURJ_{n-k,k\lfloor cn/k\rfloor})\\
 & \geq\deg_{3^{-k}}(\SURJ_{k(\lfloor n/k\rfloor -1),k\lfloor cn/k\rfloor})\\
 & \geq\deg_{3^{-k}}(\AND_k\circ \SURJ_{\lfloor n/k\rfloor -1,\lfloor cn/k\rfloor})\\
 & \geq k\deg_{1/3}(\SURJ_{\lfloor n/k\rfloor-1,\lfloor cn/k\rfloor})\\
 & \geq k\cdot\tilde\Omega\left(\frac{n}{k}\right)^{3/4}\\
 & =\tilde\Omega(n^{3/4}k^{1/4}),
\end{align*}
where the first step
uses~(\ref{eq:surj-extend-domain-and-range});
the second and third steps
use~(\ref{eq:surj-extend-domain});
the fourth step applies the claim from the opening
paragraph of the proof; the fifth step
is valid by Proposition~\ref{prop:tensor-deg-eps};
and the sixth step invokes Theorem~\ref{thm:surj-bounded}.
This settles the theorem.
\end{proof}

\subsection{Permutation testing}

We now turn to the permutation testing problem, which
requires a more subtle analysis than the functions that we
have examined so far. The difficulty is that
permutation testing does not admit a self-reduction with
$\AND$ as an outer function. To address this, we will need to
generalize Proposition~\ref{prop:tensor-deg-eps}
appropriately.
For a real $0\leq\alpha<1$ and an integer $k\geq1,$ we define $\AND_{k,\alpha}$
to be the restriction of $\AND_{k}$ to inputs whose Hamming weight
is either $k$ or at most $\alpha k.$ The following
result subsumes Proposition~\ref{prop:tensor-deg-eps} as the special case
$\alpha=(k-1)/k.$

\begin{prop}
\label{prop:tensor-deg-eps-restriction}Fix a real number $0\leq\alpha<1$
and an integer $k\geq1.$ Then for any $\epsilon\geq0$ and any function
$f\colon X\to\zoo$ on a finite subset $X$ of Euclidean space,
\begin{align*}
\deg_{\epsilon^{k}/\binom{k-1}{\lfloor\alpha k\rfloor}}(\AND_{k,\alpha}\circ f) & \geq(\lfloor\alpha k\rfloor+1)\deg_{\epsilon}(f).
\end{align*}
In particular,
\[
\deg_{(\epsilon/2)^{k}}(\AND_{k,\alpha}\circ f)\geq\alpha k\deg_{\epsilon}(f).
\]
\end{prop}

\begin{proof}
We may assume that $\deg_\epsilon(f)\ne0$ since the
proposition is trivial otherwise. 
Let $\psi$ be an $\epsilon$-error dual polynomial for $f$,
as guaranteed by Fact \ref{fact:duality}:
\begin{align*}
 & \langle f,\psi\rangle>\epsilon\|\psi\|_{1},\\
 & \orth\psi=\deg_{\epsilon}(f).
\end{align*}
Abbreviate $\ell=\lfloor\alpha k\rfloor$ and define $\Psi\colon X^{k}\to\Re$
by
\[
\Psi(x_{1},x_{2},\ldots,x_{k})=\prod_{i=1}^{k}\psi(x_{i})\cdot\prod_{i=\ell+1}^{k-1}(f(x_{1})+f(x_{2})+\cdots+f(x_{k})-i).
\]
Observe that $\Psi$ is supported on the domain of $\AND_{k,\alpha}\circ f.$
Moreover, we have the pointwise inequality
\begin{align}
|\Psi| & \leq|\psi^{\otimes k}|\prod_{i=\ell+1}^{k-1}i\nonumber \\
 & =|\psi^{\otimes k}|\cdot\frac{(k-1)!}{\ell!}.\label{eq:Psi-psi-pointwise}
\end{align}
Now
\begin{align*}
\langle\Psi,\AND_{k,\alpha}\circ f\rangle & =\langle\Psi,f^{\otimes k}\rangle\\
 & =(k-\ell-1)!\,\langle\psi^{\otimes k},f^{\otimes k}\rangle\\
 & >(k-\ell-1)!\,\epsilon^{k}\|\psi\|_{1}{}^{k}\\
 & =(k-\ell-1)!\,\epsilon^{k}\|\psi^{\otimes k}\|_{1}\\
 & \geq\epsilon^{k}\cdot\frac{(k-\ell-1)!\,\ell!}{(k-1)!}\,\|\Psi\|_{1}\\
 & =\epsilon^{k}\binom{k-1}{\ell}^{-1}\|\Psi\|_{1},
\end{align*}
where the next-to-last step uses~(\ref{eq:Psi-psi-pointwise}). Applying Fact
\ref{fact:duality} once again,
\begin{align*}
\deg_{\epsilon^{k}/\binom{k-1}{\ell}}(\AND_{k,\alpha}\circ f) & \geq\orth\Psi\\
 & \geq(\ell+1)\orth\psi\\
 & =(\ell+1)\deg_{\epsilon}(f). &  & \qedhere
\end{align*}
\end{proof}

For $m\leq n,$
a permutation testing instance $\phi\colon\{1,2,\ldots,m\}\to\{1,2,\ldots,m\}$
can be extended in a natural way to a larger instance $\Phi\colon\{1,2,\ldots,n\}\to\{1,2,\ldots,n\}$
by letting $\Phi(i)=i$ for $i=m+1,m+2,\ldots,n.$ This gives
\begin{align}
\deg_{\epsilon}(\PTP_{m,\alpha}) & \leq\deg_{\epsilon}\left(\PTP_{n,\frac{m}{n}\cdot\alpha+\frac{n-m}{n}}\right), &  & m\leq n.\label{eq:ptp-subproblem}
\end{align}
We are now in a position to prove our lower
bound on the $\epsilon$-approximate degree of
permutation testing. 

\begin{thm} \label{ptpthm}
Let $0<\alpha<1$ be a given constant. Then 
\begin{align}
\deg_{\epsilon}(\PTP_{n,\alpha}) & =\Omega\left(n^{1/3}\left(\log\frac{1}{\epsilon}\right)^{2/3}\right), &  & \frac{1}{3^{n}}\leq\epsilon\leq\frac{1}{3}.\label{eq:ptp-eps-deg}
\end{align}
\end{thm}
\begin{proof}
Let $0<\beta<1$ be arbitrary. We claim that for any positive integers
$n$ and $k$ with $k\mid n,$ the permutation testing function $\PTP_{n,\beta}$
contains 
\begin{equation}
\AND_{k,\beta/2}\circ\PTP_{n/k,\beta/2}\label{eq:and-ptp}
\end{equation}
as a subfunction. The proof is similar to that for element distinctness.
Specifically, view instances of~(\ref{eq:and-ptp}) as block-diagonal
matrices with $k$ blocks of size $n/k$ each. Then a positive instance
of~(\ref{eq:and-ptp}) is a permutation matrix and therefore a positive
instance of $\PTP_{n,\beta}$. A negative instance of~(\ref{eq:and-ptp}),
on the other hand, features at least $k-\frac{\beta}{2}k$ blocks
from $(\PTP_{n/k,\beta/2})^{-1}(0)$ and therefore corresponds to
a mapping $\{1,2,\ldots,n\}\to\{1,2,\ldots,n\}$ with a range of size
at most
\[
n-\left(k-\frac{\beta k}{2}\right)\cdot\left(\frac{n}{k}-\frac{\beta n}{2k}\right)\leq\beta n.
\]
In particular, any negative instance of~(\ref{eq:and-ptp}) is also
a negative instance of~$\PTP_{n,\beta}$. This completes the proof
of the claim. 

Now for any $\epsilon\geq0$ and any $k\in\{1,2,\ldots,\lceil\alpha n/2\rceil\},$
we have
\begin{align}
\deg_{\epsilon}(\PTP_{n,\alpha}) & \geq\deg_{\epsilon}(\PTP_{k\lfloor n/k\rfloor,\alpha/2})\nonumber \\
 & \geq\deg_{\epsilon}(\AND_{k,\alpha/4}\circ\PTP_{\lfloor n/k\rfloor,\alpha/4}),\label{eq:ptp-tensor-subproblems}
\end{align}
where the first inequality uses~(\ref{eq:ptp-subproblem}), and the
second inequality follows from the claim established in the previous
paragraph. The rest of the proof is analogous to those for $\AND_{n}$
and $\ED_{n}$. 
For $k=1,2,\ldots,\lceil\alpha n/2\rceil,$ 
\begin{align*}
\deg_{6^{-k}}(\PTP_{n,\alpha}) & \geq\deg_{6^{-k}}(\AND_{k,\alpha/4}\circ\PTP_{\lfloor n/k\rfloor,\alpha/4})\\
 & \geq\frac{\alpha k}{4}\deg_{1/3}(\PTP_{\lfloor n/k\rfloor,\alpha/4})\\
 & =\frac{\alpha k}{4}\cdot\Omega\left(\frac{n}{k}\right)^{1/3}\\
 & =\Omega(n^{1/3}k^{2/3}),
\end{align*}
where the first three steps are valid by (\ref{eq:ptp-tensor-subproblems}), Proposition~\ref{prop:tensor-deg-eps-restriction},
and Corollary~\ref{cor:ptp-bounded}, respectively.  This directly
implies~(\ref{eq:ptp-eps-deg}). 
\end{proof}

We will now show that Theorem~\ref{ptpthm} is optimal
with respect to all parameters.  In fact, we will prove
the stronger result that permutation testing has an $\epsilon$-error
quantum query algorithm with cost
$O(n^{1/3}\log^{2/3}(1/\epsilon))$.  Our quantum
algorithm is inspired by the well-known algorithm for the
collision problem due to Brassard et al.~\cite{bht}.

\begin{thm}
\label{thm:ptp-qquery-upper}
Let $0<\alpha<1$ be a given constant. Then for all
$n\geq1$ and $1/3^n\leq\epsilon\leq1/3,$
the permutation testing problem $\PTP_{n,\alpha}$ has
an $\epsilon$-error quantum query algorithm with cost
$O(n^{1/3}\log^{2/3}(1/\epsilon)).$
In particular,
\begin{align}
\deg_\epsilon(\PTP_{n,\alpha})=O\left(n^{1/3}
\left(\log\frac1\epsilon\right)^{2/3}\right).
\label{eq:ptp-eps-upper}
\end{align}
\end{thm}
\begin{proof}
We give an algorithm whose only quantum component is
Grover search. Specifically, we will only use the
fact that, given query
access to $N$ items of which $M$ are marked, Grover
search finds a marked item with probability $2/3$ using
$O(\sqrt{N/M})$ queries (see, e.g.,~\cite{bht,brassard1998quantum}). 
We will follow the convention in the quantum query
literature 
and view the input to $\PTP_{n,\alpha}$ as a function
$\phi\colon\{1,2,\ldots,n\}\to\{1,2,\ldots,n\},$ where
the algorithm has query access to $\phi.$ 

Let $s$ be an integer parameter to be determined later. Our
algorithm starts by choosing a uniformly random subset
$S\subseteq\{1,2,\ldots,n\}$ of cardinality $|S|=s.$
Next, we query $\phi$ at every point of $S.$ If $\phi$
is not one-to-one on $S,$ we output ``false.''
In the complementary case, 
we execute Grover search $\log(1/\epsilon)$ times independently,
each time looking for a point $i\in\overline S$ with the
property that $\phi(i)\in\phi(S).$ We output ``false''
if such a point is found, and ``true'' otherwise.

If $\phi$ is a permutation, the described algorithm is
always correct.  In the complementary case when $|\im\phi|\leq\alpha n,$ 
there are at least $(1-\alpha)n$ points
$i\in\{1,2,\ldots,n\}$ such that $|\phi^{-1}(\phi(i))|\geq 2.$ 
Call such points \emph{special.}
We will henceforth assume that $S$ contains at least
$(1-\alpha)s/2$
special points, which happens with probability 
at least $1-\exp(-\Theta_\alpha(s)).$ If $\phi$ is not
one-to-one on $S,$ the algorithm correctly outputs
``false.'' If $\phi$ is one-to-one on $S$ and $S$
contains at least $(1-\alpha)s/2$ special points, then 
each of the Grover executions has $\geq(1-\alpha)s/2$
eligible points to output from among a total of
$|\overline S|=n-s$ possibilities; this means that each 
Grover execution finds an eligible point with probability at
least $2/3$ using $O(\sqrt{n/((1-\alpha)s)})$
queries, thereby forcing the correct output. In
summary, the described algorithm 
has error probability at most
$\exp(-\Theta_\alpha(s))+(1/3)^{\log(1/\epsilon)}$ and
query cost $s + O(\sqrt{n/((1-\alpha)s)}\cdot\log(1/\epsilon)).$ 
In particular, error $\epsilon$ can be achieved with
query cost $O(n^{1/3}\log^{2/3}(1/\epsilon)).$
This query bound in turn
implies~(\ref{eq:ptp-eps-upper}) using the standard
transformation of a quantum query algorithm to a
polynomial; see, e.g., Ambainis~\cite{ambainis05collision}.
\end{proof}

\section{QMA Lower Bounds}
The objective of this section is to ``lift'' the approximate degree lower bound of
Theorem~\ref{ptpthm} to QMA query complexity.
As our first step, we generalize our lower bound to one-sided approximation.
The \emph{one-sided} $\eps$-approximate degree of a function $f \colon X \to \Re$, denoted
$\deg^+_{\eps}(f)$, is the least degree of a polynomial $p \colon X \to \Re$
such that 
$|p(x)| \leq \eps$ for all $x\in f^{-1}(0),$ and
$p(x) \geq  1-\epsilon$ for all $x\in f^{-1}(1).$
Thus, $p$ approximates $f$ uniformly on $f^{-1}(0)$ but may 
take on arbitrarily large values on $f^{-1}(1).$ It is clear from the definition
that $\deg^+_\eps(f)\leq\deg_\eps(f).$
The gap between these quantities can be large in general, such as $1$ versus $\Omega(\sqrt n)$ 
for the bounded-error approximation of $\OR_n$. However, we will show that
these two notions of approximation are equivalent for the permutation testing function.
\begin{prop}
\label{prop:ptp-onesided-twosided}
For all $\alpha,\epsilon,$ and $n,$
\begin{align}
\deg^+_\epsilon(\PTP_{n,\alpha})&= \deg_\epsilon(\PTP_{n,\alpha}).
\label{eq:onesided-twosided-ptp}
\end{align}
\end{prop}
\noindent
This equality of approximate degree and one-sided approximate degree for permutation 
testing has the important consequence that the lower bound of Theorem~\ref{ptpthm}
applies to the one-sided setting as well. The proof of Proposition~\ref{prop:ptp-onesided-twosided}
is based on the observation that any one-sided approximant for permutation testing
can be symmetrized to be constant on $f^{-1}(1),$ effectively making it a two-sided
approximant. This technique was used previously in~\cite[Theorem 2]{BT15} to argue
that $\deg^+_\epsilon(\ED_n)=\deg_\epsilon(\ED_n).$

\begin{proof}[Proof of Proposition~\emph{\ref{prop:ptp-onesided-twosided}}.]
Let $p$ be a one-sided approximant for $\PTP_{n,\alpha}$ with error $\epsilon,$ 
so that $|p|\leq \epsilon$ on $\PTP_{n,\alpha}^{-1}(0)$ and
$p\geq 1-\epsilon$ on $\PTP_{n,\alpha}^{-1}(1).$ Define 
\begin{align}
p^*(x)=\Exp p(\sigma x\tau),
\label{eq:p-star}
\end{align}
where $\sigma,\tau$ are uniformly random permutations on $\{1,2,\ldots,n\},$ and $\sigma x\tau$
denotes the matrix obtained by permuting the rows of $x$ according to $\sigma$ and the columns according to $\tau.$
Then $p^*$ is also a one-sided approximant for $\PTP_{n,\alpha}$ 
because $\PTP_{n,\alpha}^{-1}(0)$ 
and $\PTP_{n,\alpha}^{-1}(1)$ are closed under permutations of rows and columns. 
Moreover, $p^*$ takes on the
same value, call it $M$, at all 
$x\in\PTP_{n,\alpha}^{-1}(1)$ because $\sigma x\tau$ in (\ref{eq:p-star}) 
is a uniformly random permutation matrix in that case. As a result, the
normalized polynomial $p^*/\max\{1,M\}$ approximates $\PTP_{n,\alpha}$ pointwise within $\epsilon.$ 
Finally, $\deg p^* \leq \deg p$ because $p^*$ is an average of polynomials,
each obtained from $p$ by permuting the input variables. 
\end{proof}

We will also need the following proposition, implicit in Marriott and Watrous's
proof~\cite{marriott2005quantum} of Vyalyi's result \cite{vyalyi2003qma} on
$\QMA$ and $\SBQP.$ For completeness, we include its short proof. 

\begin{prop} 
\label{prop:qma-to-poly}
Suppose that $f \colon X \to \zoo$ has a QMA query protocol with witness length $m$
and query cost $q$. Then there is a polynomial $p \colon X \to \Re$ such that
\begin{align}
&\deg p=O(mq),\label{eq:p-deg}\\
&|p(x)| \leq 2^{-2m}&&\text{for all $x \in f^{-1}(0),$}
\label{eq:p-false}\\
&p(x) \geq 2^{-m-1}&& \text{for all $x \in f^{-1}(1).$} 
\label{eq:p-true}
\end{align}
\end{prop}
\begin{proof}
Marriott and Watrous \cite{marriott2005quantum}
showed that the soundness and completeness errors of the QMA query protocol for $f$ can be driven down to $2^{-2m}$
without an increase in witness length, and with only a factor of $O(m)$ increase in query cost.
This yields a QMA protocol $\mathcal{Q}$ for $f$ that has witness length $m$, query cost $O(m q)$, and soundness and completeness errors $2^{-2m}$.
That is, on any input in $f^{-1}(1)$, there exists a witness that causes Arthur to accept with probability
at least $1-2^{-2m}$, and on any input in $f^{-1}(0)$, for every witness that might be sent by Merlin, Arthur accepts with probability
at most $2^{-2m}$. 

Now run $\mathcal{Q}$ with the witness fixed to the totally mixed state. This yields
a quantum query algorithm $\mathcal{A}$.
On inputs in $f^{-1}(0)$, the acceptance probability of $\mathcal{A}$ is at most the soundness
error of $\mathcal{Q}$, which is at most $2^{-2m}$. On inputs in $f^{-1}(1)$ , the acceptance probability 
of $\mathcal{A}$ is at least $(1-2^{-2m}) \cdot 2^{-m} \geq 2^{-m-1}$. 
Now (\ref{eq:p-deg})--(\ref{eq:p-true}) follow from the well-known result of Beals et al.~\cite{beals} that the acceptance probability of any $T$-query quantum algorithm on input $x$ is a polynomial $p(x)$ of degree at most $2T$. 
\end{proof}

We have reached our main result on the QMA complexity of
permutation testing, stated as Theorem~\ref{thm:qma} in the introduction.
For the reader's convenience, we restate the theorem here.

\begin{thm*}
Let $0<\alpha<1$ be an arbitrary constant. Then any QMA query 
protocol for $\PTP_{n,\alpha}$ with 
witness length $m$ has query cost $\Omega(n/m)^{1/3}$.
In particular, $\PTP_{n,\alpha}$ has QMA complexity $\Omega(n^{1/4}).$
\end{thm*}
\begin{proof}
Fix a QMA query protocol for $\PTP_{n,\alpha}$ with witness length $m\in[3,n]$
and query cost $q.$ Then Proposition~\ref{prop:qma-to-poly} gives a 
polynomial $p$ satisfying (\ref{eq:p-deg})--(\ref{eq:p-true}). It follows
that $2^{m+1}p$ approximates $\PTP_{n,\alpha}$ in a one-sided manner
to error $2^{-m+1},$ forcing
$\deg^+_{2^{-m+1}}(\PTP_{n,\alpha})=O(mq).$
On the other hand, taking $\epsilon=2^{-m+1}$ in Theorem~\ref{ptpthm} and Proposition~\ref{prop:ptp-onesided-twosided} 
shows that $\deg^+_{2^{-m+1}}(\PTP_{n,\alpha})=\Omega(n^{1/3}m^{2/3}).$
Comparing these complementary bounds on the one-sided approximate degree of
permutation testing gives
$q=\Omega(n/m)^{1/3}$ and thus $\max\{m,q\}=\Omega(n^{1/4}).$
\end{proof}

We remind the reader that by virtue of Proposition~\ref{prop:ptp-ptp},
the variant of permutation testing studied in this
paper is equivalent to Aaronson's permutation testing problem~\cite{aaronson}. 
As a result, Theorems~\ref{thm:qma},
\ref{eq:ptp-eps-deg}, and~\ref{thm:ptp-qquery-upper}
and Proposition~\ref{eq:onesided-twosided-ptp}
remain valid with $\PTP_{n,\alpha}$ replaced by 
$\PTP^*_{n,\alpha}$.

\section{Open Problems}
A natural next step would be to close the gap between our $\Omega(n^{1/4})$ QMA lower bound for permutation testing 
and the known upper bound of $O(n^{1/3})$. 
In addition, we highlight the well-known open question of resolving the QMA communication complexity of set disjointness.
The best known lower bound here is $\Omega(n^{1/3})$ \cite{klauck}, while the best upper bound is $O(n^{1/2})$. 
We believe that both questions highlight significant gaps in our understanding of $\QMA.$
Another natural open question is whether the na\" ive error-reduction method for approximate degree is optimal.
Namely, it is well known that $\deg_{\eps}(f) \leq O( \min\{\deg_{1/3}(f)\log(1/\epsilon), n\})$
for every $f\colon\zoon\to\zoo,$ yet this bound is not known to be tight for any such $f$ with
sublinear approximate degree. 
It \emph{is} tight for some $f$ whose domain is a proper subset of $\zoon,$
based for example on approximate counting.

\bibliographystyle{alpha}
\newcommand{\etalchar}[1]{$^{#1}$}

\end{document}